\documentclass{llncs}
\usepackage{latexsym,amsmath,amssymb,url}

\newcommand{\NP}{\text{\normalfont NP}}
\renewcommand{\P}{\text{\normalfont P}}

\let\phi=\varphi

\renewenvironment{proof}{\begin{pf}}{\qed\end{pf}}

\begin{document}

\title{Vertex Cover Problem Parameterized Above and Below Tight Bounds}

\author{Gregory Gutin\inst{1} \and Eun Jung Kim\inst{1} \and Michael Lampis\inst{2}\and Valia Mitsou\inst{2}}
\institute{  Royal Holloway, University of London, UK\\
  \email{\{gutin|eunjung\}@cs.rhul.ac.uk}\and City University of New York, USA\\
  \email{\{mlampis|vmitsou\}@gc.cuny.edu}}

\date{ }
\maketitle

\begin{abstract}
We study the well-known Vertex Cover problem parameterized above and below tight bounds.
We show that two of the parameterizations (both were suggested by Mahajan, Raman and Sikdar, J. Computer and System Sciences, 75(2):137--153, 2009)
are fixed-parameter tractable and two other parameterizations are W[1]-hard (one of them is, in fact, W[2]-hard).
 \end{abstract}

\pagenumbering{arabic}
\pagestyle{plain}

\section{Introduction}

A parameterized problem $\Pi$ can be considered as a set of pairs
$(I,k)$ where $I$ is the \emph{main part} and $k$ (usually an integer)
is the \emph{parameter}. $\Pi$ is called \emph{fixed-parameter
tractable} (FPT) if membership of $(I,k)$ in $\Pi$ can be decided by a {\em fixed-parameter} algorithm,
i.e., an algorithm of running
time $O(f(k)|I|^{O(1)})$, where $|I|$ denotes the size of $I$ and $f(k)$ is a
computable function
(for further background and terminology on parameterized complexity we
refer the reader to the
monographs~\cite{DowneyFellows99,FlumGrohe06,Niedermeier06}). If the
nonparameterized version of $\Pi$ (where $k$ is just a part of input)
is $\NP$-hard, then the function $f(k)$ must be superpolynomial
provided $\P\neq \NP$. Often $f(k)$ is ``moderately exponential,''
which makes the problem practically feasible for small values of~$k$.
Thus, it is important to parameterize a problem in such a way that the
instances with small values of $k$ are of real interest.

For a graph $G=(V,E)$, a set $C\subseteq V$ is a {\em vertex cover}
if for every edge $uv\in E$ at least one of the vertices $u,v$
belongs to $C$. The well-known classical {\sc Vertex Cover} problem
is the problem of deciding whether a graph $G$ has a vertex cover of
size at most $k.$ This problem is $\NP$-complete and the {\em standard} parameterization
considers $k$ as the (only) parameter.

{\sc Vertex Cover} was named the {\em Drosophila} of fixed-parameter algorithmics
\cite{DowneyFellows99,GuoNieWer07,Niedermeier06,NiedermeierRossmanith00}
as (i) there is a long list of improvements on the exponential function in $k$ of
fixed-parameter algorithms with the currently best exponential bound being below $1.28^k$ \cite{ChenKanjXia05},
(ii) it has applications in various areas including bioinformatics and linear programming
\cite{Cheetham03,DowneyFellows99,FlumGrohe06,GutKarRaz09,Niedermeier06}, (iii) {\sc Vertex Cover} has been a benchmark for developing sophisticated data
reduction and problem kernelization techniques \cite{Abu04}, (iv) research on the problem has led us to new research directions
within parameterized complexity such as counting \cite{ArvRam02}, enumerating \cite{Dam06,Fernau02} and parallel processing
\cite{Abu06,Cheetham03}.

\section{Nonstandard Parameterizations of Vertex Cover}\label{sec:nst}

In the above-mentioned applications of {\sc Vertex Cover}, $k$ is relatively small, but this is not the case in
some other applications. Indeed, let $B$ be a positive integer and consider
the family ${\cal G}_B$ of graphs with all degrees bounded from
above by $B$ \cite{Esperet09,Grigorieff06}. The vertex cover in a
graph from ${\cal G}_B$ on $m$ edges must have at least $m/B$
vertices. However, $m/B$ is not small for $m$ large enough and,
thus, the standard parameterization of {\sc Vertex Cover} is of
little interest for ${\cal G}_B$.

Mahajan, Raman and Sikdar \cite{MahajanRamanSikdar09} observed that
$m/B$ is a tight\footnote{Indeed, consider the disjoint union of
$m/B$ stars $K_{1,B}$.} lower bound on the the minimum cardinality
of a vertex cover of a graph in ${\cal G}_B$  and stated the following
parameterized problem.

\begin{quote}
  {\bfseries Vertex Cover Above Tight Lower Bound-1} (VCL1)\\
  \emph{Instance:} A positive integer $B$, a graph $G=(V,E)\in {\cal G}_B$, a positive integer $k$.\\
  \emph{Parameters:} $k$ and $B$.\\
  \emph{Question:} Is there a vertex cover $C$ of $G$ with at most $m/B+k$ vertices?
\end{quote}

In fact, Mahajan, Raman and Sikdar \cite{MahajanRamanSikdar09}
stated VCL1 as a problem with one parameter, $k$, with $B$ being a
constant. They asked whether the one parameter version of VCL1 is
fixed-parameter tractable. In Section \ref{sec:LB} we prove that
VCL1 is FPT, which implies that the one parameter version of VCL1 is
FPT as well. Since VCL1 has two parameters rather than the usual one, our result can be viewed as a contribution
towards Multivariate Algorithmics as outlined by Fellows \cite{FelIWOCA}.

A variation of {\sc Vertex Cover} which has been studied extensively in
the literature \cite{GuoNieWer07,DomLOkSauVil08} is {\sc Capacitated Vertex Cover}, where every vertex of
the graph has a given capacity which is a limit on the number of its
incident edges it can cover. The $m/B$ lower bound also applies in
this case and we can define the following problem.

\begin{quote}
  {\bfseries Capacitated Vertex Cover Above Tight Lower Bound-1} (CVCL1)\\
  \emph{Instance:} A positive integer $B$, a graph $G=(V,E)\in {\cal G}_B$, a positive integer $k$.\\
  \emph{Parameter:} $k$.\\
  \emph{Question:} Is there a vertex cover $C$ of $G$ with at most $m/B+k$ vertices?
\end{quote}

The standard parameterization of {\sc Capacitated Vertex Cover} is FPT
\cite{GuoNieWer07,DomLOkSauVil08}, but surprisingly CVCL1 turns out to be W[2]-hard as we prove in Section \ref{sec:UB}.

Let $G=(V,E)\in {\cal G}_B$. Observe that the chromatic number $\chi(G)$ of $G$
is at most $B+1$ (one can properly color the vertices of $G$ in at most $B+1$
colors using the greedy algorithm \cite{West01}). Thus, $G$ has an independent
set $I$ of size at least $n/(B+1)$, where $n=|V|.$ Since $V\setminus I$ is a
vertex cover, $G$ has a vertex cover of size at most $nB/(B+1)$. Observe that
$nB/(B+1)$ is a tight upper bound on the minimum size of a vertex cover of $G$
since the disjoint collection of $t$ copies of $K_{B+1}$'s must be covered by
at least $tB$ vertices. Mahajan, Raman and Sikdar \cite{MahajanRamanSikdar09}
formulated the following problem (in fact, its one-parameter version 
with $k$ being the parameter and $b$ being a constant).

\begin{quote}
  {\bfseries Vertex Cover Below Tight Upper Bound-1} (VCU1)\\
  \emph{Instance:} A positive integer $B$, a graph $G=(V,E)\in {\cal G}_B$, and a positive integer $k$.\\
  \emph{Parameter:} $k$ and $B$.\\
  \emph{Question:} Is there a vertex cover $C$ of $G$ with at most $nB/(B+1)-k$ vertices?
\end{quote}

It is easy see that VCU1 is FPT using Brooks' Theorem
\cite{West01}: $\chi(G)\le B$ unless one of the connectivity components of $G$
is $K_{B+1}$ or, if $B=2$ and one of the connectivity components of $G$ is an 
odd cycle. If $\chi(G)\le B$, we have that the
vertex cover of $G$ is at most $n(B-1)/B$.  If $n(B-1)/B < nB/(B+1) -k$ we can
trivially answer {\sc yes}, otherwise we have $ kB(B+1)\ge n$ which gives a problem
kernel.  Thus, to make VCU1 more interesting, one should replace $nB/(B+1)-k$
with $n(B-1)/B-k$. We do not know the parameterized complexity of this
modification of VCU1.

Let us note in passing that if $B$ is not considered a parameter or a constant
but rather a part of the input (that is, if we study the problem on graphs of
unbounded degree) and $k$ is the only parameter, then this version of VCU1 is
W[1]-hard.  To see this, suppose that we are given an instance of the standard
parameterization of {\sc Independent Set} consisting of a graph $G=(V,E)$ on $n$
vertices containing at least one edge and the aim is to check the existence of an independent set of size $k$ 
(or, a vertex cover of size $n-k$).  
Obtain a new graph $H$ by adding to $G$ a pair $x,y$ of new vertices  and 
connecting $x$ to all vertices of $G$ and $y$. Observe that $H$ has $N=n+2$ vertices and
the maximum degree $B$ of $H$ is $N-1$. Thus, the problem for $H$ is whether $H$ has a vertex cover
of size $NB/(B+1)-k=n+1-k$. Since any minimum vertex cover of $H$ must contain $x$ and must not contain 
$y$, the problem for $H$ is equivalent to asking whether $G$ has a vertex cover of size $n-k$. A similar argument
can be applied if we parameterize below $N(B-1)/B$, and, thus, these
parameterizations of vertex cover are mainly of interest for graphs of bounded
degree.

Mahajan, Raman and Sikdar \cite{MahajanRamanSikdar09} have also mentioned the following
parameterization of {\sc Vertex Cover} whose parameterized complexity was determined recently.
This problem is of interest since $\mu$ is a tight lower bound on the minimum size of a vertex cover.

\begin{quote}
{\bfseries Vertex Cover Above Tight Lower Bound-2} (VCL2)\\
    \emph{Instance:} A graph $G=(V,E)$ with a maximum matching of size $\mu$ and a positive integer $k$.\\
  \emph{Parameter:} $k$.\\
  \emph{Question:} Is there a vertex cover $C$ of $G$ with at most $\mu+k$ vertices?
\end{quote}

The parameterized complexity of VCL2 remained open for quite some
time until recently Razgon and O'Sullivan \cite{Razgon08} proved
that {\sc Min 2-Sat Deletion} is FPT and since VCL2 is
fixed-parameter reducible to {\sc Min 2-Sat Deletion}
\cite{Mishra07}, VCL2 is also FPT. (In {\sc Min 2-Sat Deletion}, we are given a CNF formula $F$ with $m$ clauses such that each clause has two literals and asked whether there is a truth assignment that satisfies at least $m-k$ clauses; $k$ is the parameter.)

If $\mu'$ is the size of a {\em maximal} matching $M$ of a graph $G$, $G$ has a vertex cover with $2\mu'$ vertices, the set of vertices of $M.$ Observe that $2\mu'$ is a tight upper bound on the minimum size of a vertex cover\footnote{Consider the disjoint union of $P_3$'s.}. Thus, the following problem is another natural parameterization of {\sc Vertex Cover}.

\begin{quote}
  {\bfseries Vertex Cover Below Tight Upper Bound-2} (VCU2)\\
  \emph{Instance:} A graph $G=(V,E)$, a maximal matching $M$ of $G$ and a positive integer $k$.\\
  \emph{Parameter:} $k$.\\
  \emph{Question:} Is there a vertex cover $C$ of $G$ with at most $2|M|-k$ vertices?
\end{quote}

Unfortunately, VCU2 is W[1]-hard as we show in Section \ref{sec:UB}.

\section{VCL1 is FPT}\label{sec:LB}

In this section we assume that the graph $G$ under consideration is in ${\cal G}_B$, where $B$ is positive integral constant.
The following proposition characterizes VCL1 instances $\{G=(V,E),k=0,B\}$ for which the answer is {\sc yes}. This proposition allows us to check whether $G$ has a vertex cover of size $m/B$ in polynomial time.

\begin{proposition} \label{indep}
A graph $G$ has a vertex cover of size exactly $m/B$ if and only if $G$ is a bipartite graph with no isolated vertices and with one partite set of size $m/B$.
\end{proposition}
\begin{proof}
If $G$ is bipartite with no isolated vertices and with one partite set of size $m/B$, then this partite set is clearly a vertex cover. Suppose $G=(V,E)$ has a vertex cover $C$ of size $m/B$. Since a vertex can cover at most $B$ edges, every vertex of $C$ covers exactly $B$ edges and no two of them cover the same edge. Therefore $C$ forms an independent set and thus $G$ is a bipartite graph with bipartite sets $C$ and $V-C$. Clearly, $G$ has no isolated vertices.
\end{proof}

\vspace{2mm}

Now we will prove that VCL1 is fixed-parameter tractable.

\begin{lemma} \label{edgebip}
If the answer to a VCL1 instance is {\sc yes}, then there exists a set $D$ of at most $kB$ edges whose deletion makes $G$ bipartite.
\end{lemma}
\begin{proof}
Suppose $C$ is a vertex cover of $G$ with at most $m/B+k$ vertices
and let $D\subseteq E$ be a set of edges with both endvertices in
$C$. Then we have $(m/B+k)B\geq m+|D|$ and thus $|D|\leq kB$.
Since no two vertices of $C$ cover the same
edge in $G-D$, $G-D$ has no odd cycle implying that
$G-D$ is bipartite.
\end{proof}

\vspace{3mm}

If deleting the edges of $D\subseteq E$ makes $G$ bipartite, we say
$D$ is an {\em edge bipartization}. Due to Lemma \ref{edgebip}, we
may assume that the input graph $G$ has an {\em edge bipartization}
of size at most $kB$. The problem of deciding whether $G$ has an
edge bipartization of size at most $p$ is known as the {\sc Edge
Bipartization} problem. This problem has been studied extensively and
the best known fixed-parameter algorithm runs in time $O(2^pm^2)$, see
\cite{Guo2006}.

Since we know that there is a small edge bipartization of $G$,
we can solve {\sc Vertex Cover} optimally as follows. Suppose $D$ is an edge
bipartization of $G$ with $|D|\leq kB$,
$C$ is a vertex cover of $G$, and let $C^*=C\cap V_D$, where
$V_D$ is the set of endvertices of edges of $D$. Note that $C^*$ is a
vertex cover of $G[D]$ and let $C'\subseteq C^*$ be a minimal vertex cover of $G[D]$.
Observe that $C\setminus C'$ is a vertex cover of the bipartite graph $G-C'$.
Thus, a vertex cover of $G$ is the union of a minimal vertex cover $C'$ of $G[D]$ and
a vertex cover of the bipartite graph $G-C'$.

Consider the following procedure $\Pi$ to generate vertex covers of $G[D]$:
pick either $u$ or $v$ from every edge $uv\in D$. We call all vertex covers generated by $\Pi$, $\Pi$-{\em vertex covers}.
To see that the set of $\Pi$-vertex covers includes all minimal vertex covers of $G[D]$
observe that if $uv\in D$ and $\{u,v\}\subseteq C'$, where $C'$ is a
minimal vertex cover of $G[D]$, then a neighbor $x$ of $u$ is not
in $C'$. Thus, $u$ can be picked up by $\Pi$ while considering the edge $ux$ and $v$
can be picked up by $\Pi$ while considering the edge $uv.$

Our algorithm proceeds as follows. First, we find an edge
bipartization $D$ of $G$ such that $|D|\leq kB$ using the algorithm
of \cite{Guo2006}. If such an edge bipartization $D$ does not exist,
the answer to VCL1 is {\sc no}. Otherwise, we use $\Pi$ to generate all $\Pi$-vertex
covers of $G[D]$ and for each such vertex cover $C'$, we find a
minimum-size vertex cover $C''$ of the bipartite graph $G-C'$ and
check whether $|C'|+|C''|\le m/B+k.$

Let us evaluate the running time of this algorithm. We can find $D$ in time $O(2^{kB}m^2).$
Clearly, all $\Pi$-vertex covers of $G[D]$ can be generated in $O(2^{kB}n^{O(1)})$ time.
For a vertex cover $C'$, we can find a minimum-size vertex cover $C''$ of the bipartite graph $G-C'$ in time $O(m^2\sqrt{n}).$
Thus, the total running time of our algorithm is $O(2^{kB}n^{O(1)}).$

Thus, we have obtained the following result.

\begin{theorem}
The problem VCL1 can be solved in time $O(2^{kB}n^{O(1)}).$
\end{theorem}

\section{W[2] and W[1]-hardness Results}\label{sec:UB}

\begin{theorem}
CVCL1 is W[2]-hard.
\end{theorem}
\begin{proof}
We give a parameterized reduction from {\sc Dominating Set}: Given a graph
$G=(V,E)$ with $|V|=n$ and maximum degree $B$, we are asked whether it
has a dominating set of size $k$. We will construct a graph
$G'=(V',E')$ with maximum degree $B+2$ such that $G$ has a dominating
set of size $k$ if and only if $G'$ has a vertex cover of size $|E'|/(B+2)+k$.

For every vertex $v\in V$ construct a choice gadget which will be a
complete bipartite graph $K_{B+1,B+2}$. Also, for every vertex $v\in
V$ construct a domination gadget, which is simply a vertex $v_d$
with $B-d(v)+1$ leaves attached to it, where $d(v)$ is the degree of
$v$ in $G$. Finally, for each vertex $v$ and every $u\in N[v]$,
where $N[v]$ is the closed neighborhood of $v$ in $G$, add an edge
from a different vertex of the larger partite set of $v$'s choice gadget to
the vertex $u_d$ in $u$'s domination gadget. This completes the
construction of $G'$.

The capacities of all the vertices of $G'$ are equal to their
degrees, except for the vertices $v_d$ of the domination gadgets, to
which we give capacities $B+1$, which is one less than their degree.

First, let us calculate the lower bound for this graph. In every
choice gadget we have $(B+1)(B+2)$ edges. Also, there are $d(v)+1$
edges connecting every $v_d$ with the choice gadgets and $B-d(v)+1$
edges connecting it to the attached leaves, so there are $B+2$ edges
incident on each $v_d$. Therefore, in total we have
$|E'|=n(B+1)(B+2)+n(B+2)=n(B+2)^2$. The maximum degree is $B+2$,
therefore we are looking for a vertex cover of size $n(B+2)+k$.

Suppose that the original graph has a dominating set $D$ of size
$k$. In $G'$ we select the following vertices in the vertex cover:
from each choice gadget corresponding to a vertex in $D$ we select
the larger partite set of the bipartite graph ($B+2$ vertices) and we select
the smaller partite set of the bipartite graph ($B+1$ vertices) from all the
other choice gadgets. We also select all vertices $v_d$. In total we
have selected $n(B+1)+k+n=n(B+2)+k$ vertices. It is not hard to see
that this is indeed a vertex cover. It is also a capacitated vertex
cover because the only vertices constrained by the capacities are
the vertices $v_d$. However, because $D$ is a dominating set every
$v_d$ is connected to a choice gadget from which we picked the larger
partite set, thus for every $v_d$ one of its incident edges is covered from
our selection in the choice gadgets and $v_d$ has enough capacity to
cover all of its remaining incident edges.

Now for the converse, suppose that $G'$ has a capacitated vertex
cover of size $n(B+2)+k$. First, note that without loss of
generality we may assume that this vertex cover includes all
vertices $v_d$, because if such a vertex is not in the cover at
least one leaf is in the cover and we can simply exchange the two.
Also, without loss of generality no leaves are in the cover, because
we have already argued that the neighbors of the leaves are in the
cover. Therefore, the only reason to include a leaf may be that the
cover includes some $v_d$ but none of its neighbors in the choice
gadgets, thus exceeding
$v_d$'s capacity. However, in such a case we can exchange the leaf
with one of $v_d$'s neighbors from the choice gadgets, thus covering
at least as many edges without exceeding $v_d$'s capacity. Finally,
in each choice gadget the cover includes either the $B+1$ vertices
of the smaller partite set or the $B+2$ vertices of the larger partite, because
if a single vertex of the larger partite set is in the cover it is clearly
optimal to take all vertices of the larger partite set and none from the
smaller partite set. So we can conclude that there exist $k$ choice gadget
where we have picked all of the larger partite set and we have picked the
smaller partite set from all the others. Now, if this is indeed a capacitated
vertex cover, every $v_d$ has a neighbor in the cover so that its
capacity is not exceeded and as argued previously this neighbor is
from the choice gadgets. Therefore, if we select in $G$ the $k$
vertices which correspond to the $k$ gadgets where we picked the
larger partite set they must form a dominating set.
\end{proof}

\begin{theorem}
VCU2 is W[1]-hard.
\end{theorem}
\begin{proof}
We give a parameterized reduction from {\sc Independent Set} to VCU2. Let $G=(V,E)$ with parameter $k$ be an instance of {\sc Independent Set}. Let $n=|V|$ and $m=|E|$. We shall construct a new graph $G'=(V',E')$ as follows. We replace each vertex $u\in V$ by a path $u_1u_2u_3$ of length two. For each edge $e=uv\in E$, we build a gadget graph $G_e=(V_e,E_e)$ with $V_e=\{e_u,e_v,e_w,e_z\}$ and $E_e=\{e_ue_v,e_ve_w,e_we_u,e_we_z\}$. Finally let $E'$ have edges $u_3e_u$ and $v_3e_v$ for each edge $e=uv\in E$ plus all the edges mentioned above. We set $M=\{u_2u_3:u\in V\}\cup \{e_ue_w:e=uv\in E\}.$ Observe that $M$ is a maximal matching of $G'$. (In fact, it is a maximal matching of $G'$ with minimum number of edges.) In the following, we show that there is a vertex cover of size at most $n-k$ in $G$ (i.e., $G$ has an independent set with at least $k$ vertices) if and only if there is a vertex cover of size at most $2|M|-k$ in $G'$. Note that $|M|=n+m$ and $2|M|-k=2n+2m-k$.

First, suppose $C$ is a vertex cover in $G$ with at most $n-k$ vertices. Then for $C'\subseteq V'$ we pick up the following vertices: (1) $u_2$ for every $u\in V$, (2) $e_w$ for every $e=uv\in E$, (3) if $u\in C$ and $v\notin C$ for edge $e=uv \in E$, choose $u_3$ and $e_v$, (4) if $u,v\in C$ for edge $uv\in E$, choose $u_3,v_3$ and $e_u$. Observe that $C'$ covers every edge of $G'$ and contains exactly $n+m+m+|C|\leq 2|M|-k$ vertices.

Second, if we have a vertex cover $C'$ of $G'$, then we know a
vertex cover of size at most $|C'|$ can be obtained by replacing a
vertex of degree 1 by its neighbor. So we may assume that
$\{u_2:u\in V\} \cup \{e_w:e=uv\in E\}$ belong to $C'$. The edges of
$G'$ not covered by these vertices are a path $u_3e_u,e_ue_v,e_vv_3$
of length two for every edge $e=uv\in E$; we denote the subgraph of
$G$ induced by all such paths by $G''$. Suppose $C'\leq 2|M| -k =
2n+2m-k$ and let $C''=C''_m\cup C''_v$ be a vertex cover in $G''$ of
size at most $n+m-k$, where $C''_m=C''\cap \{e_u,e_v:e=uv\in E\}$
and $C''_v=C''\cap \{u_3:u\in V\}$. Since at least one of $e_u$ and
$e_v$ must be in $C''$ in order to cover the edge $e_ue_v$, we have
$|C''_m|\geq m$ and thus $|C''_v|\leq n-k$. Moreover, edges covered
by $\{e_u,e_v\}$ can be also covered by either $\{u_3,e_v\}$ or
$\{e_u,v_3\}$. Hence we may assume that $|C''_m|=m$. Now observe
that $C''_v$ includes at least one of $u_3$ and $v_3$ for every edge
$e=uv\in E$ so as to cover those edges in $G''-C''_m$. Therefore
$\{u\in V:u_3\in C''_v\}$ is a vertex cover of $G$.
\end{proof}

\section{Further Research}

We could not determine the parameterized complexity of the next problem mentioned in Section \ref{sec:nst}.
It is likely that VCU is FPT. At least, this follows for the special case of the problem when the size of a maximum clique in $G$
is sufficiently smaller than $B,$ see \cite{Reed98}.

\begin{quote}
  {\bfseries Vertex Cover Below Tight Upper Bound} (VCU)\\
  \emph{Instance:} A positive integer $B$, a graph $G=(V,E)\in {\cal G}_B$, and a positive integer $k$.\\
  \emph{Parameter:} $k$ and $B$.\\
  \emph{Question:} Is there a vertex cover $C$ of $G$ with at most $n(B-1)/B-k$ vertices?
\end{quote}

Regarding the parameterization of vertex cover above $m/B$ we have established
that the problem is FPT when both $k$ and $B$ are considered parameters, and
also that when only $k$ is considered a parameter and the maximum degree is
unbounded the capacitated version of the problem is W[2]-hard. A natural
question left open is what happens to the uncapacitated version of the
problem in the second case.

\subsection*{Acknowledgments} Research of Gutin and Kim was
supported in part by an EPSRC grant.

%\bibliographystyle{abbrv}
%\bibliography{stefan}

\end{document}